\documentclass[a4paper,twocolumn,11pt,accepted=2023-10-30]{quantumarticle}
\pdfoutput=1
\usepackage[utf8]{inputenc}
\usepackage[english]{babel}
\usepackage[T1]{fontenc}
\usepackage{amsmath}
\usepackage{hyperref}

\usepackage{tikz}
\usepackage{lipsum}

\usepackage[numbers,sort&compress]{natbib}

\usepackage{graphicx,framed} 
\usepackage{dcolumn} 
\usepackage{bm} 
\usepackage{braket}
\usepackage{color}
\usepackage[normalem]{ulem} 
\usepackage{qcircuit}       

\usepackage{amsthm,amssymb}

\newtheorem{theorem}{Theorem}
\usepackage{comment}

\definecolor{lightblue}{RGB}{73,151,208}
\definecolor{crimson}{RGB}{140,41,53}

\hypersetup{
    colorlinks,
    linkcolor={crimson},
    citecolor={lightblue},
    urlcolor={lightblue}
}

\newtheorem{corollary}{Corollary}[theorem]

\newcommand{\hl}[1]{\textcolor{black}{#1}}

\usepackage{mathtools}

\begin{document}

\title{Minimum Trotterization Formulas for a Time-Dependent Hamiltonian}

\author{Tatsuhiko N. Ikeda}
\affiliation{RIKEN Center for Quantum Computing, Wako, Saitama 351-0198, Japan}
\affiliation{Department of Physics, Boston University, Boston, Massachusetts 02215, USA}
\affiliation{Institute for Solid State Physics, University of Tokyo, Kashiwa, Chiba 277-8581, Japan}
\orcid{0000-0001-6135-1726}

\author{Asir Abrar}
\affiliation{Physics and Informatics Laboratory, NTT Research, Inc.,940 Stewart Dr., Sunnyvale, California, 94085, USA}
\author{Isaac L. Chuang}
\affiliation{Department of Physics, Department of Electrical Engineering and Computer Science, and Co-Design Center
for Quantum Advantage, Massachusetts Institute of Technology, Cambridge, Massachusetts 02139, USA
}

\author{Sho Sugiura}
\affiliation{Physics and Informatics Laboratory, NTT Research, Inc.,940 Stewart Dr., Sunnyvale, California, 94085, USA}
\affiliation{Laboratory for Nuclear Science, Massachusetts Institute of Technology, Cambridge, 02139, MA, USA}

\maketitle

\begin{abstract}
When a time propagator $e^{\delta t A}$ for duration $\delta t$ consists of two noncommuting parts $A=X+Y$, Trotterization approximately decomposes the propagator into a product of exponentials of $X$ and $Y$. 
Various Trotterization formulas have been utilized in quantum and classical computers, but much less is known for the Trotterization with the time-dependent generator $A(t)$.
Here, for $A(t)$ given by the sum of two operators $X$ and $Y$ with time-dependent coefficients $A(t) = x(t) X + y(t) Y$, we develop a systematic approach to derive high-order Trotterization formulas with minimum possible exponentials.
In particular, we obtain fourth-order and sixth-order Trotterization formulas involving seven and fifteen exponentials, respectively, which are no more than those for time-independent generators.
We also construct another fourth-order formula consisting of nine exponentials having a smaller error coefficient.
Finally, we numerically benchmark the fourth-order formulas in a Hamiltonian simulation for a quantum Ising chain, showing that the 9-exponential formula accompanies smaller errors per local quantum gate than the well-known Suzuki formula.
\end{abstract}

\section{Introduction}
Let us consider an initial value problem
\begin{align}
&\frac{d}{dt} S(t,t') = A(t) S(t,t'), \label{eq: diff eq}
\\
&S(t',t')=I,
\end{align} 
where $A(t), S(t,t') \in \mathbb{C}^{d\times d}$ and $I$ is the identity operator in $\mathbb{C}^{d\times d}$.
Equation~\eqref{eq: diff eq} frequently appears in physics.
For instance, the Schr\"{o}dinger equation corresponds to the case where $A(t)$ is anti-Hermitian (i.e., $A(t)=-iH(t)$ with a Hermitian Hamiltonian $H(t)$). The Lindblad equation, which is the time evolution equation for Markovian open systems, corresponds to the case where $A(t)$ is the Liouvillian, which is not Hermitian nor anti-Hermitian in general. Therefore, solving Eq.~\eqref{eq: diff eq} is one of the fundamental problems in physics.

The formal solution of Eq.~\eqref{eq: diff eq} is given by the following time-ordered exponential,
\begin{align}\label{eq:Udef}
S(t,t')=\mathcal{T}\exp\left( \int_{t'}^t A(s)ds \right).
\end{align}
When the operator $A(t)$ is time-independent, i.e., $A(t)=A$ $\forall t$, Eq.~\eqref{eq:Udef} reduces to $S(t,t')=\exp((t-t')A)$.
The time-ordered exponential is defined as an infinite sum of multiple integrals and is not easy to calculate since $A(t)$ consists of non-commuting operators in general. Thus, a decomposition that expresses \eqref{eq:Udef} as a product of exponentials of simpler operators is important in applications such as a digital quantum simulation.

For the time-independent case of $A=X+Y$, \hl{decompositions of $S(\delta t,0)=e^{(X+Y)\delta t}$ into a product of $e^{\alpha X}$ and $e^{\beta Y}$ have been extensively studied.} 
The lowest-order decomposition $S(\delta t,0)= e^{\delta t X}e^{\delta t Y}+O(\delta t^2)$ was first obtained by Lie and its applicability was extended by Trotter~\cite{Trotter1959} and Kato~\cite{Kato1974}.
Also, various higher-order decompositions have been proposed both in classical mechanics \cite{Forest1990} and in quantum mechanics \cite{Hatano2005}.
Here, Trotterization means the form of
\begin{align}
  T(\vec{a},\vec{b})=
  \begin{cases}
      e^{a_1 X}e^{b_1 Y}\cdots e^{a_q X}e^{b_q Y}\\
      e^{a_1 X}e^{b_1 Y}\cdots e^{a_q X}e^{b_q Y}e^{a_{q+1}X}.
  \end{cases}
  \label{eq: product ansatz}
\end{align}
with real numbers $a_j$'s and $b_j$'s, \hl{where $2q$ or $2q+1$ denotes the number of exponentials in each formula}. 
When \hl{$a_j$'s and $b_j$'s are chosen so that} the error of the approximation is of the order of $\delta t^{p+1}$, i.e., 
\begin{align}
    S\left(\mu+\frac{\delta t}{2},\mu-\frac{\delta t}{2}\right)=T(\vec{a},\vec{b}) + O(\delta t^{p+1}), 
\end{align}
we call $T(\vec{a},\vec{b})$ the $p$-th order Trotterization formula of $S(\mu+\frac{\delta t}{2},\mu-\frac{\delta t}{2})$.
\hl{Unlike other methods~\cite{Low2018}, Trotterization allows us to implement Schr\"{o}dinger equation in digital quantum computers without exploiting ancillary qubits and has been widely used in noisy intermediate-scale quantum computers.}

However, much less is known about Trotterization for the time-dependent case~\cite{Hatano2005,Huyghebaert1990,Poulin2011,Low2018}, \hl{which includes more variety of nonequilibrium quantum phenomena.} 
For simplicity, we focus, throughout this paper, on the following type of a time-dependent generator~\cite{An2021}
\begin{align}
    A(t) = x(t) X + y(t) Y,
    \label{eq: time dependence}
\end{align}
where $x(t)$ and $y(t)$ are smooth real functions and $X, Y \in \mathbb{C}^{d\times d}$.
\hl{
Two different approaches are known for the Trotterization formula of a time-dependent operator. The first is to take $x(t)$ and $y(t)$ at fine-tuned discrete points $t_k\ (n=1, \cdots, N)$, and the algorithm is constructed using $\exp(a_k X\delta t)$ and $\exp(b_k y(t_k)Y)$~\cite{Hatano2005,Wiebe2010}.
This approach systematically generates higher-order formulas based on the formulas for time-independent cases, but its efficiency, especially the number of exponentials, is not necessarily optimal.
The second is an approach based on a different approximation initiated by Huyghebaert and de Readt~\cite{Huyghebaert1990,Poulin2011}, where the time-ordered exponential~\eqref{eq:Udef} is approximated by a product of normal exponentials such as $\exp\left(\int_{t'}^t x(s)X ds \right)$ and $\exp\left(\int_{t'}^t y(s)Y ds \right)$.
However, its generalization to higher orders is not straightforward and has not been systematically developed.
}

\hl{Here we focus on the number of exponentials in the formula as its usefulness}. The number is directly related to the calculation cost for, e.g., a digital quantum computer and a matrix-product-state-(MPS-)based classical simulation. Therefore, we look for the Trotter formula of the minimum number of exponentials. 
For \hl{$k=2$}, the well-known formula so-called midpoint rule (see Eq.~\eqref{eq:midpoint} below) is a three-exponential formula and is the minimum formula. This is because it is readily shown that there does not exist a 2-exponential second-order formula in general.
For $k=4$, by contrast, there are some known formulas.
For example, Suzuki~\cite{Suzuki1993,Hatano2005,Wiebe2010} gave a 15-exponential formula for the more general case $X(t)+Y(t)$, and this reduces to an 11-exponential one for the present class of problem given in Eq.~\eqref{eq:Udef}.
However, it has remained unsolved whether there exists another fourth-order formula having fewer exponentials, and the minimum number of exponentials has not been found yet, to the best of the authors' knowledge.

In this paper, \hl{we derive minimum Trotterization formulas for time-dependent} $A(t)$ in the form of Eq.~\eqref{eq: time dependence}.
\hl{In particular, we obtain the fourth-order (sixth-order) explicit formulas consisting of 7 (15) exponentials and show that 7 (15) exponentials are the minimum among all the fourth-order (sixth-order) formulas in the form of Eq.~\eqref{eq: product ansatz}.
Thus we call our formulas the minimum fourth-order and sixth-order Trotterization (MFT and MST) formulas.}
We also find a fourth-order 9-exponential formula, whose error is less than that of the MFT.
We numerically compare the error between the MFT, 9-exponential, and Suzuki's fourth-order formulas in 1-qubit and many-spin models.
We confirm that the MFT (9-exponential) formula has a slightly larger (smaller) error per exponential than Suzuki's does.
In this way, the 9-exponential formula offers an efficient way for simulating dynamics on, e.g., a digital quantum computer and an MPS-based classical simulator.

\section{The formulas and errors}
We begin with the Magnus expansion~\cite{Blanes2009}
\begin{align}\label{eq:Magnus}
S\left(\mu+\frac{\delta t}{2},\mu-\frac{\delta t}{2}\right) = \exp\left(\sum_{n=1}^\infty \Omega_n\right).
\end{align}
As is well known, $\Omega_n$ is given by a multiple integral over $n$ variables $t_1,\dots,t_n$ on the domain $\mu-\delta t/2\le t_n \le \cdots \le t_1 \le \mu+\delta t/2$, and the integrand consists of nested commutators of $A(t_1),\dots,A(t_n)$  (see Appendix~\ref{app:Magnus} for their concrete forms).
Therefore, $\Omega_1=O(\delta t)$ and $\Omega_n=O(\delta t^{n+1})$ ($n\ge2$) hold in general.
Also, our parameterization ensures that $\Omega_n$ does not involve even-order terms in $\delta t$ when Taylor-expanded.
The Magnus expansion, truncated at some order, has been utilized for classical and quantum simulations (see, e.g., Refs.~\cite{Iserles1999,Sornborger2018}).
Here we use it to derive Trotterization formulas, i.e., decompositions into $X$ and $Y$ as in Eq.~\eqref{eq: product ansatz}.

For $A(t)$ given in Eq.~\eqref{eq: time dependence}, the Magnus expansion reduces to the continuous BCH formula~\cite{Blanes2009} giving
\begin{align}
\Omega_1 &= \beta_1(\mu,\delta t)X + \beta_2(\mu,\delta t)Y,\\
\Omega_2 &= \beta_{12}(\mu,\delta t)[X,Y],\\
\Omega_3 &= \sum_{i=1}^2 \beta_{i12}(\mu,\delta t)[Z_i,[X,Y]],\\
\Omega_4 &= \sum_{i=1}^2 \sum_{j=1}^2 \beta_{ij12}(\mu,\delta t)[Z_i,[Z_j,[X,Y]]],
\end{align}
and so on, where we introduced useful notations $Z_1=X$ and $Z_2=Y$, and $\beta$'s are real numbers, \hl{such as 
$\beta_1=\int_{\mu-\delta t/2}^{\mu+\delta t/2}dt x(t)$ and 
$\beta_2=\int_{\mu-\delta t/2}^{\mu+\delta t/2}dt y(t)$
(see Appendix~\ref{app:Magnus} for the others).}
For brevity, we shall omit $(\mu,\delta t)$ from $\beta$'s in the following.

\hl{
We remark that Eq.~\eqref{eq:Magnus} becomes a trivial identity in the time-independent case, where $x(t)=x$ and $y(t)=y$.
In this case, the exact evolution does not need the time-ordered exponential, giving $S(\mu+\delta t/2,\mu-\delta t/2)=\exp[(xX+yY)\delta t]$.
Also, $\Omega_n=0$ holds for $n>1$ since $\beta_{12}=\beta_{112}=\dots=0$, and the right-hand side of Eq.~\eqref{eq:Magnus} equals $\exp[(xX+yY)\delta t]$, where we used $\beta_1=x\delta t$ and $\beta_2=y\delta t$.
}

Interestingly, $\Omega_n$'s are actually smaller than the obvious estimate $O(\delta t^n)$.
For instance,
$\Omega_2\propto \beta_{12}=\frac{1}{2}\int_{\mu-\delta t/2}^{\mu+\delta t/2}dt_2\int_{\mu-\delta t/2}^{t_2}dt_1 [ y(t_2)x(t_1) - x(t_2)y(t_1) ]=O(\delta t^3)$.
To prove this, one can, e.g., Taylor-expand $x(t_j)$ and $y(t_j)$ at $t_j=\mu$, finding that the coefficient of $\delta t^2$ vanishes.
Furthermore, one can show $\Omega_3 = O(\delta t^5)$ and $\Omega_4=O(\delta t^5)$ (see Appendix~\ref{app:Magnus} for proof).
This rapid increase of orders in $\Omega_n$ enables us to construct efficient Trotterization formulas as shown below.

Using the lowest-order continuous BCH formula, we can reproduce the well-known second-order Trotterization formula $S(\mu+\frac{\delta t}{2},\mu-\frac{\delta t}{2})=T_2(\mu+\frac{\delta t}{2},\mu-\frac{\delta t}{2})+O(\delta t^3)$ with
\begin{align}\label{eq:midpoint}
    T_2\left(\mu+\frac{\delta t}{2},\mu-\frac{\delta t}{2}\right) \equiv e^{x(\mu)X\delta t/2}e^{y(\mu)Y\delta t}e^{x(\mu)X\delta t/2}
\end{align}
also known as the mid-point rule.
To do this, we notice $S(\mu+\frac{\delta t}{2},\mu-\frac{\delta t}{2})=e^{\Omega_1}+O(\delta t^3)=e^{\beta_1 X+\beta_2 Y}+O(\delta t^3)$ and use the minimum second-order formula for a time-independent problem $e^{\beta_1X+\beta_2Y}=e^{\beta_1X/2}e^{\beta_2Y}e^{\beta_1X/2}+O(\delta t^3)$ together with $\beta_1=x(\mu)+O(\delta t^3)$ and that for $\beta_2$.

To obtain a fourth-order formula, we invoke the following key result.
\begin{theorem}\label{thm:main}
Suppose $u\equiv \beta_{12}/\beta_2=O(\delta t^2)$. Then it follows
\begin{align}
&S\left(\mu+\frac{\delta t}{2},\mu-\frac{\delta t}{2}\right) =e^{uX}e^{\beta_1X+\beta_2Y}e^{-uX}+\Upsilon_5;\label{eq:uX}\\
&\Upsilon_5  = \Omega_3+\Omega_4-\frac{u^2}{2}\beta_2[X,[X,Y]]+O(\delta t^7).\label{eq:error}
\end{align}
\end{theorem}
\begin{proof}
The BCH formula leads to $e^{uX}e^{\beta_1 X+\beta_2 Y}e^{-uX}=e^{\beta_1X+\beta_2Y+u[X,\beta_2Y]}+\frac{u^2}{2}[X,[X,\beta_2 Y]]+O(\delta t^7)=e^{\beta_1X+\beta_2Y+\beta_{12}[X,Y]}+\frac{u^2}{2}\beta_2[X,[X,Y]]+O(\delta t^7)$, where we used $u=O(\delta t^2)$ ensured by our assumption.
On the other hand, the continuous BCH formula gives
\begin{align}
S\left(\mu+\frac{\delta t}{2},\mu-\frac{\delta t}{2}\right)&=e^{\beta_1X+\beta_2Y+\beta_{12}[X,Y]}\notag\\
&\qquad+\Omega_3+\Omega_4+O(\delta t^7) \label{eq:Legendre}
\end{align}
Comparing these, we obtain Eq.~\eqref{eq:uX}.
\end{proof}

\hl{Three remarks are in order regarding Theorem~\ref{thm:main}. The first remark is}
on the role of the assumption $\beta_{12}/\beta_2=O(\delta t^2)$. Although $\beta_2=O(\delta t)$ always holds true, $\beta_2$ can happen to be as small as 
$\beta_2=\Theta(\delta t^3)$. Here, $\beta=\Theta(\delta t^n)$ stands for $0<\lim_{\delta t\to +0}(\beta/\delta t^n)<\infty$, and we used the fact that $\beta_2$ consists of odd-order terms of $\delta t$ (see Appendix~\ref{app:poly}). In such a case, if $\beta_{12}=\Theta(\delta t^3)$, $u=\Theta( 1)$ and Eq.~\eqref{eq:uX} does not hold.
However, even when $\beta_{12}/\beta_1 \neq O(\delta t^2)$, if $\beta_{12}/\beta_2=O(\delta t^2)$, we can still use Theorem~\ref{thm:main} by interchanging the roles of $x(t)X$ and $y(t)Y$.
Recall that our target $S(\mu+\frac{\delta t}{2},\mu-\frac{\delta t}{2})$ is invariant under the interchange of $x(t)X$ and $y(t)Y$.
If both $\beta_{12}/\beta_i$ $(i=1,2)$ are not $O(\delta t^2)$, Theorem~\ref{thm:main} can be used by setting $u=0$ as follows.
This case implies $\beta_1=\beta_2=O(\delta t^3)$ and $\Omega_n=O(\delta t^5)$ for $n\ge2$, and $S(\mu+\delta t/2,\mu-\delta t/2)=e^{\beta_1X+\beta_2Y}+O(\delta t^5)$.

\hl{The second remark is} that Theorem~\ref{thm:main} can also be derived using  (a non-Hermitian version of) the Schrieffer-Wolff Transformation (SWT) \cite{Schrieffer66, Bravyi11}.
When we regard $\Omega_1$ and $\Omega_2$ as the unperturbed and perturbation terms, respectively, the SWT aims to eliminate $\Omega_2$ within $O(\delta t^5)$ errors by an appropriate similarity transformation ($S$ is not necessarily Hermitian below)
\begin{align}\label{eq: SWT}
    e^{S}(\Omega_1+\Omega_2)e^{-S} =  \Omega_1  +O(\delta t^5),
\end{align}
whose exponential form reads
\begin{align}\label{eq: SWTexp}
    e^{S}e^{\Omega_1+\Omega_2}e^{-S} =  e^{\Omega_1}  +O(\delta t^5).
\end{align}
Assuming that $S=O(\delta t^2)$, we obtain the condition for $S$ satisfying Eq.~\eqref{eq: SWT} as
\begin{align}\label{eq: Scomm}
    [S,\Omega_1] + \Omega_2 = 0
\end{align}
when $O(\delta t^5)$ errors are neglected.
Recalling that $\Omega_1 = \beta_1 X+ \beta_2 Y$ and $\Omega_2 = \beta_{12}[X,Y]$, we obtain a solution $S=(\beta_{12}/\beta_2)X=uX$ for Eq.~\eqref{eq: Scomm}, for which Eq.~\eqref{eq: SWTexp} coincides with Eq.~\eqref{eq:uX}.
We remark that the transformation is unitary if $X$ is Hermitian.

\hl{
The final remark is that Theorem~\ref{thm:main} becomes trivial for the time-independent case, where we have $u=0$ and $\Upsilon_5$.
This follows from $\Omega_n=0$ for $n>1$, as remarked above.
In this sense, $\Upsilon_5$ is the characteristic of time-dependent cases, which we will call the time-dependent component of the Trotter error.
}

\hl{Having given those three remarks, we now return to the Trotterization formula's derivation.}
Even though we are considering a time-dependent problem,
Eq.~\eqref{eq:uX} tells us that the Trotterization is achieved by Trotterizing $e^{\beta_1X+\beta_2Y}$, which is done by invoking a conventional formula for time-independent problems.
For example, we can apply the Forest-Ruth-Suzuki formula~\cite{Forest1990,Suzuki1990}
\begin{align}
e^{\beta_1 X+\beta_2Y}&=e^{\frac{s \beta_1}{2} X}e^{s \beta_2 Y}e^{\frac{1-s}{2}\beta_1 X}e^{(1-2s)\beta_2 Y}\notag\\
&\qquad\times e^{\frac{1-s}{2}\beta_1 X}e^{s \beta_2 Y}e^{\frac{s \beta_1}{2} X}
+\Gamma_5,\label{eq:Forest-Ruth}
\end{align}
where where $s=(2-2^{1/3})^{-1}$ and $\Gamma_5=O(\delta t^5)$ is the error~\cite{Omelyan2002,Ostmeyer2023}, whose properties are generally investigated~\cite{Childs2021}.
Substituting Eq.~\eqref{eq:Forest-Ruth} into Eq.~\eqref{eq:uX}, we obtain the following corollary giving the MFT.

\begin{corollary}[Minimum 7-exponential formula]\label{coro:MFT}
It follows
\begin{align}\label{eq:MFT}
&S\left(\mu+\frac{\delta t}{2},\mu-\frac{\delta t}{2}\right)=T_{7,4}+\Lambda_5;\notag\\
&T_{7,4}\equiv e^{(\frac{s \beta_1}{2}+u) X}e^{s \beta_2 Y}e^{\frac{1-s}{2}\beta_1 X}e^{(1-2s)\beta_2 Y} \notag\\
&\qquad\qquad \times e^{\frac{1-s}{2}\beta_1 X}e^{s \beta_2 Y}e^{(\frac{s \beta_1}{2}-u) X};\notag\\
&\Lambda_5 = \Gamma_5+\Upsilon_5 .
\end{align}
\end{corollary}

Remarkably, the 7-exponential formula~\eqref{eq:MFT} is minimum among fourth-order formulas.
To prove this, we recall that our time-dependent problem involves the time-independent problems as a special case of $x(t)=x$ and $y(t)=y$.
For the special case, Eq.~\eqref{eq:MFT} reduces to the fourth-order Forest-Ruth-Suzuki formula~\cite{Forest1990,Suzuki1990} as one can check easily.
Recall that the 7-exponential fourth-order Forest-Ruth-Suzuki formula is minimum among time-independent problems.
Thus every fourth-order formula for general time-dependent problems requires at least 7 exponentials.
Theorem~\ref{thm:main} shows that such a formula actually exists despite a reasonable inference that more general problems demand more exponentials.

The error parts $\Gamma_5$ and $\Upsilon_5$ of the total error $\Lambda_5$ are time-independent and -dependent components, respectively.
As we discussed, $\Upsilon_5=0$ holds for the special case that $x(t)$ and $y(t)$ are time-independent. 
Since $\Gamma_5$ is the well-known Trotter error for time-independent case~\cite{Childs2021}, $\Upsilon_5$ is the additional error source particular to the time dependence.

We remark on the time-reversal character of the MFT.
The exact evolution satisfies
\hl{
\begin{align}
S\left(\mu-\frac{\delta t}{2},\mu+\frac{\delta t}{2}\right)^{-1} =S\left(\mu+\frac{\delta t}{2},\mu-\frac{\delta t}{2}\right),
\end{align}}
which follows from $S(\mu+\frac{\delta t}{2},\mu-\frac{\delta t}{2})S(\mu-\frac{\delta t}{2},\mu+\frac{\delta t}{2})=1$.
Nicely, the MFT $T_{7,4}$ also has this property.
\hl{Namely, if we apply Theorem~\ref{thm:main} to $S(\mu-\frac{\delta t}{2},\mu+\frac{\delta t}{2})$ by substituting $\delta t\to -\delta t$ to have
\begin{align}
&S(\mu-\frac{\delta t}{2},\mu+\frac{\delta t}{2})=T_{7,4}'+\Lambda_5',\\
&T_{7,4}'=e^{(\frac{s \beta_1'}{2}+u') X}e^{s \beta_2' Y}e^{\frac{1-s}{2}\beta_1' X}e^{(1-2s)\beta_2' Y}\notag\\ 
&\qquad\qquad\times e^{\frac{1-s}{2}\beta_1' X}e^{s \beta_2' Y}e^{(\frac{s \beta_1'}{2}-u') X},
\end{align}
then we obtain
\begin{align}\label{eq:T74_rev}
T_{7,4}'{}^{-1}=T_{7,4}.
\end{align}
The keys to confirming Eq.~\eqref{eq:T74_rev} are 
the relations $\beta_i'=-\beta_i$ $(i=1,2)$ and $\beta_{12}'=-\beta_{12}$ and their consequence $u'=u$.}
The time-reversal character~\eqref{eq:T74_rev} implies that the error is of odd order~\cite{Omelyan2002}.

Although we have substituted the minimum Forest-Ruth-Suzuki formula~\eqref{eq:Forest-Ruth} into Eq.~\eqref{eq:uX},
we can use any fourth-order Trotterization formulas instead (see, e.g.,  Ref.~\cite{Ostmeyer2023} for a list of them).
One can use another to reduce the error $\Gamma_5$ by using more exponentials than seven. 
For instance, using Omelyan's Forest-Ruth formula (see Refs.~\cite{Omelyan2002,Ostmeyer2023} for detail and real numbers $a_j$'s and $b_j$'s)
\begin{align}
e^{\beta_1 X+\beta_2 Y} &= e^{a_1 \beta_1X}e^{b_1 \beta_2 Y}e^{a_2 \beta_1 X}e^{b_2 \beta_2 Y}e^{a_3 \beta_1 X}\notag\\
&\qquad\times e^{b_2 \beta_2 Y}e^{a_2 \beta_1 X}e^{b_1 \beta_2 Y}e^{a_1 \beta_1X}+\Gamma_5'
\end{align}
in Eq.~\eqref{eq:uX},
we obtain the following corollary giving the 9-exponential fourth-order formula.
\begin{corollary}[9-exponential formula]\label{coro:9exp}
It follows
\begin{align}\label{eq:9exp}
&S\left(\mu+\frac{\delta t}{2},\mu-\frac{\delta t}{2}\right)=T_{9,4}+\Lambda_5';\notag\\
&T_{9,4}\equiv e^{(a_1 \beta_1+u)X}e^{b_1 \beta_2 Y}e^{a_2 \beta_1 X}e^{b_2 \beta_2 Y}e^{a_3 \beta_1 X}\notag\\
&\qquad\qquad\times e^{b_2 \beta_2 Y}e^{a_2 \beta_1 X}e^{b_1 \beta_2 Y}e^{(a_1 \beta_1-u)X};\notag\\
&\Lambda_5' = \Gamma_5'+\Upsilon_5.
\end{align}
\end{corollary}
\noindent Here, in exchange for adding two more exponentials, we have smaller coefficients in $\Gamma_5'$ than in $\Gamma_5$.
\hl{As we will see below in an example model in Sec.~\ref{sec:Ising}, the 9-exponential formula produces less error than the MFT and the Suzuki formula model when compared at the same number of elementary quantum gates.
}

\hl{
Before closing this section, we remark on the generalization beyond the fourth-order formula.
Recall that the MFT~\eqref{eq:MFT} reduces to the Forest-Ruth-Suzuki for time-independent cases when we set $u=0$ or eliminate the SWT.
This implies that the MFT can be derived by applying an appropriate SWT to the time-independent formula.
This approach is promising to derive higher-order minimum Trotterization formulas and, in fact, leads to the following minimum sixth-order Trotterization (MST) formula $T_{15,7}$ consisting of 15 exponentials (see Appendix~\ref{app:higher} for the derivation).
\begin{theorem}[Minimum sixth-order 15-exponential formula]\label{thm:MST}
Suppose that $\beta_1=\Theta(\delta t)$ and $\beta_2=\Theta(\delta t)$.
Then it follows
\begin{align}\label{eq:MST}
&S\left(\mu+\frac{\delta t}{2},\mu-\frac{\delta t}{2}\right)=T_{15,7}+\Lambda_7;\quad \Lambda_7=O(\delta t^7); \notag\\
&T_{15,7}\equiv e^{(a_1 \beta_1+u_4)X}e^{(b_1 \beta_2 +u_3)Y}e^{(a_2 \beta_1+u_2) X}e^{(b_2 \beta_2 +u_1 -z)Y}\notag\\
&\qquad\quad \times e^{(a_3 \beta_1-w) X}
 e^{(b_3 \beta_2 +z)Y}e^{(a_4 \beta_1 +w) X}e^{b_4 \beta_2 Y}\notag\\
 &\qquad\quad \times e^{(a_4 \beta_1+w) X}e^{(b_3 \beta_2 +z)Y}e^{(a_3 \beta_1 -w)X} e^{(b_2 \beta_2-u_1-z) Y}\notag\\
&\qquad\quad \times e^{(a_2 \beta_1-u_2) X}e^{(b_1 \beta_2 -u_3)Y}e^{(a_1 \beta_1-u_4)X},
\end{align}
where $a_1,\dots,a_4$ and $b_1,\dots,b_4$ are given in Eq.~\eqref{eq:abvals_Yoshida}, $u_1,\dots,u_4$ of $O(\delta t^2)$ and $w$ and $z$ of $O(\delta t^3)$ are given by Eqs.~\eqref{eq:ceq1} and \eqref{eq:ceq2}.
\end{theorem}
}

\hl{
We make some remarks on Eq.~\eqref{eq:MST}.
First, as in the MFT, Eq.~\eqref{eq:MST} reduces to Yoshida's minimum sixth-order formula~\cite{Yoshida1990} in the time-independent case, where $u_1=\dots=u_4=w=z=0$.
Namely, these variables are introduced as a time-dependent generalization.
Second, the 15 exponentials in Eq.~\eqref{eq:MST} are placed symmetrically around the central $e^{b_4\beta_2 Y}$, except for the asymmetric appearance of $\pm u_j$ $(j=1,\dots,4)$ due to the SWTs.
Third, $w$ and $z$ appearing symmetrically are not SWTs but are newly introduced at the sixth order, enabling $T_{15,7}$ to reproduce $\Omega_3$ in the Magnus expansion.
These structures would be generalized to even higher orders: The SWTs and the symmetric decorations like $w$ and $z$ would produce the minimum Trotterization formulas for time-dependent cases out of those for time-independent ones.
Finally, like the MFT, the MST is a minimum sixth-order formula because it has the same number of exponentials as Yoshida's formula does for time-independent generators.
}

\section{algorithm for unitary dynamics}
The arguments in the previous sections have been so general that $A(t)$ may or may not be anti-Hermitian.
Nonetheless, the case where $A(t)$ is anti-Hermitian is particularly important because it includes the Schr\"{o}dinger equation, for which $A(t)=-iH(t)$ with $H(t)$ being a Hamiltonian.
In this section, we focus on this case, summarizing the algorithm for the Hamiltonian simulation.
Considering Eq.~\eqref{eq: time dependence}, we focus on the following Hamiltonian
\begin{align}\label{eq:Hamiltonian}
H(t)=f(t)F+g(t)G,
\end{align}
where $f(t)$ and $g(t)$ are real, and $F$ and $G$ are noncommuting Hermitian operators.
To address these problems, one can utilize the above results with the replacements $A(t)\to -iH(t)$, $x(t)\to f(t)$, $X\to -iF$, etc.
As relevant applications, the MFT offers an efficient implementation in quantum dynamics simulation on a digital quantum computer and matrix-product-state(MPS-)based classical simulations.

\begin{figure}
\includegraphics[width=\columnwidth]{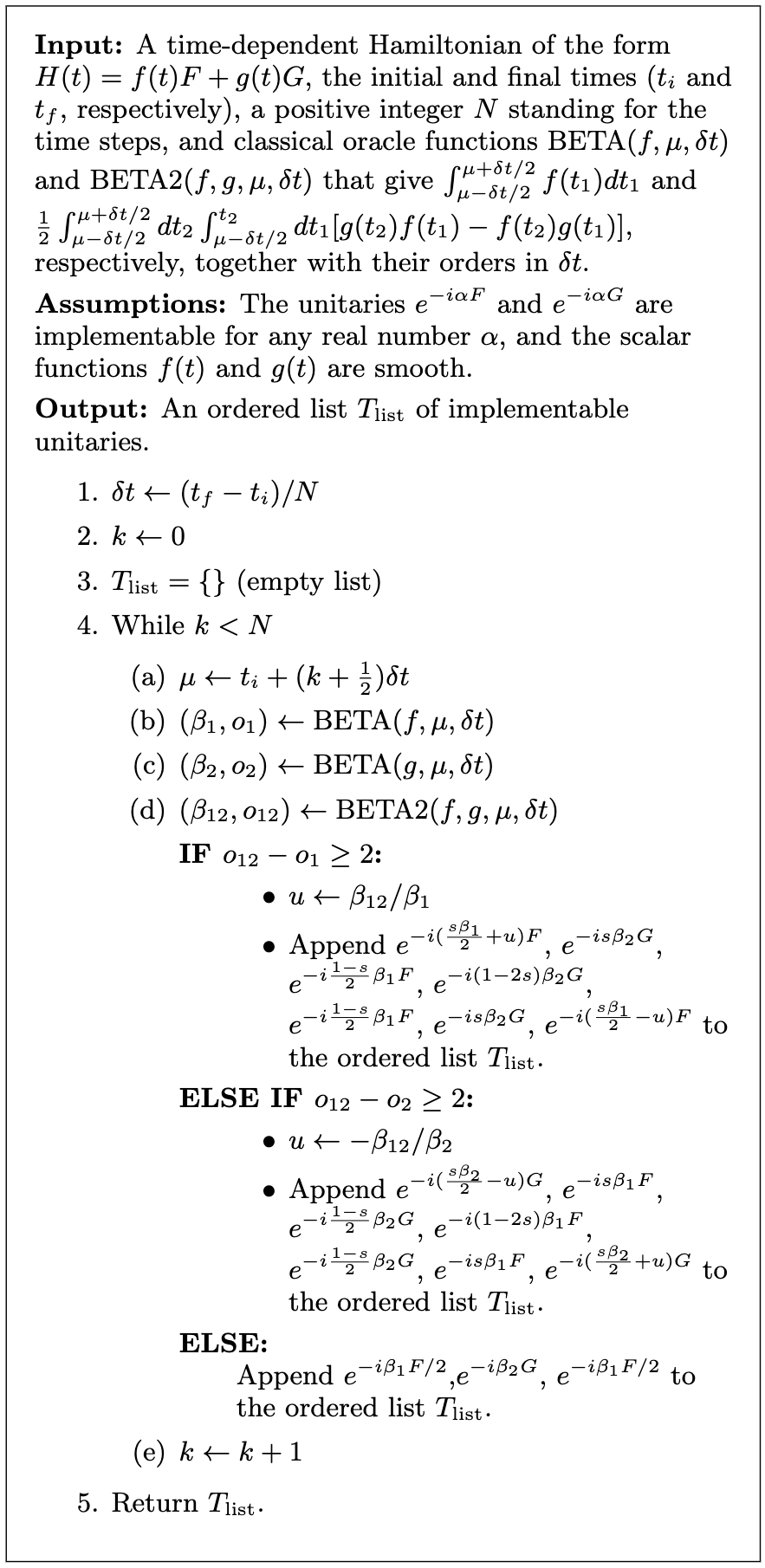}
    \caption{Pseudocode for the Hamiltonian simulation based on the MFT.}
    \label{fig:pseudocode}
\end{figure}

For usage, we summarize the algorithm as a pseudocode in Fig.~\ref{fig:pseudocode}.
\hl{Here we focus on the MFT since it works similarly with the other formulas.
In the algorithm,} we provide an ordered list of unitaries that brings an initial state $\ket{\psi(t_i)}$ at time $t=t_i$ to a final state $\ket{\psi(t_f)}$ according to the Hamiltonian~\eqref{eq:Hamiltonian} based on the MFT.
We can also use the 9-exponential formula instead of the MFT in the algorithm by replacing the seven exponentials to append with 9 exponentials appearing in Eq.~\eqref{eq:9exp}.

We note that each time step requires the oracle functions that evaluate integrals, which can be cumbersome.
In such a case, we can do it approximately up to the fourth order of $\delta t$ by using the Taylor or orthogonal-polynomial expansion for $x(t)$ and $y(t)$ (see also Appendix~\ref{app:poly}).

\section{Comparison to Other Time-Dependent Formulas}
Here we compare the MFT and 9-exponential formulas with other Trotterization formulas for time-dependent Hamiltonians~\eqref{eq:Hamiltonian}.
The known second-order formulas are equivalent to the mid-point rule~\eqref{eq:midpoint} within errors of $O(\delta t^3)$.
In fact, the second-order Suzuki~\cite{Suzuki1993,Hatano2005} is the same as the mid-point rule, and another integrator $T_\mathrm{HdR}=\mathcal{T}\exp(-i\int_{\mu}^{\mu+\delta t/2}x(s)Xds)\mathcal{T}\exp(-i\int_{\mu-\delta t/2}^{\mu+\delta t/2}y(s)Yds)$ $\times\mathcal{T}\exp(-i\int_{\mu-\delta t/2}^{\mu}x(s)Xds)$ proposed by Huyghebaert--de Readt~\cite{Huyghebaert1990} satisfies $T_\mathrm{HdR}=T_2+O(\delta t^3)$ as one can check easily.
Furthermore, those second-order 3-exponential formulas are minimum as they are known to be minimum for time-independent cases.
Thus we take Suzuki's formula (i.e., mid-point rule) as the representative for the second-order formula.

The fourth-order formula can vary in terms of error and gate complexity depending on the choice of algorithm. 
The fourth-order Suzuki formula for time-dependent Hamiltonians~\cite{Suzuki1993,Hatano2005,Wiebe2010} is made of the mid-point rule~\eqref{eq:midpoint} (with $x(t)X\to -if(t)F$ and $y(t)Y\to -ig(t)G$) as 
\begin{align}
\begin{split}
    &T_{4}\left(\mu+\frac{\delta t}{2},\mu-\frac{\delta t}{2}\right)\\
    &=T_{2}\left(\mu+\frac{\delta t}{2},\mu_4\right)
    T_{2}(\mu_4,\mu_3)\\
    &\quad \times T_{2}(\mu_3,\mu_2)
    T_{2}(\mu_2,\mu_1)
    T_{2}\left(\mu_1,\mu-\frac{\delta t}{2}\right),
\end{split}
\end{align}
where $\mu_1=\mu-\frac{1-2w}{2}\delta t$, $\mu_2=\mu-\frac{1-4w}{2}\delta t$, $\mu_3=\mu+\frac{1-4w}{2}\delta t$, $\mu_4=\mu+\frac{1-2w}{2}\delta t$, and $w =(4-4^{1/3})^{-1}$.
Note that $T_4$ consists of 11 exponentials for Hamiltonians in Eq.~\eqref{eq:Hamiltonian} because four pairs of $e^{-i\alpha F}$ ($\alpha\in\mathbb{R}$) can be combined into one between $T_2$'s.
Note again that the MFT has only 7 exponentials.

\begin{figure}
  \includegraphics[width=\columnwidth]{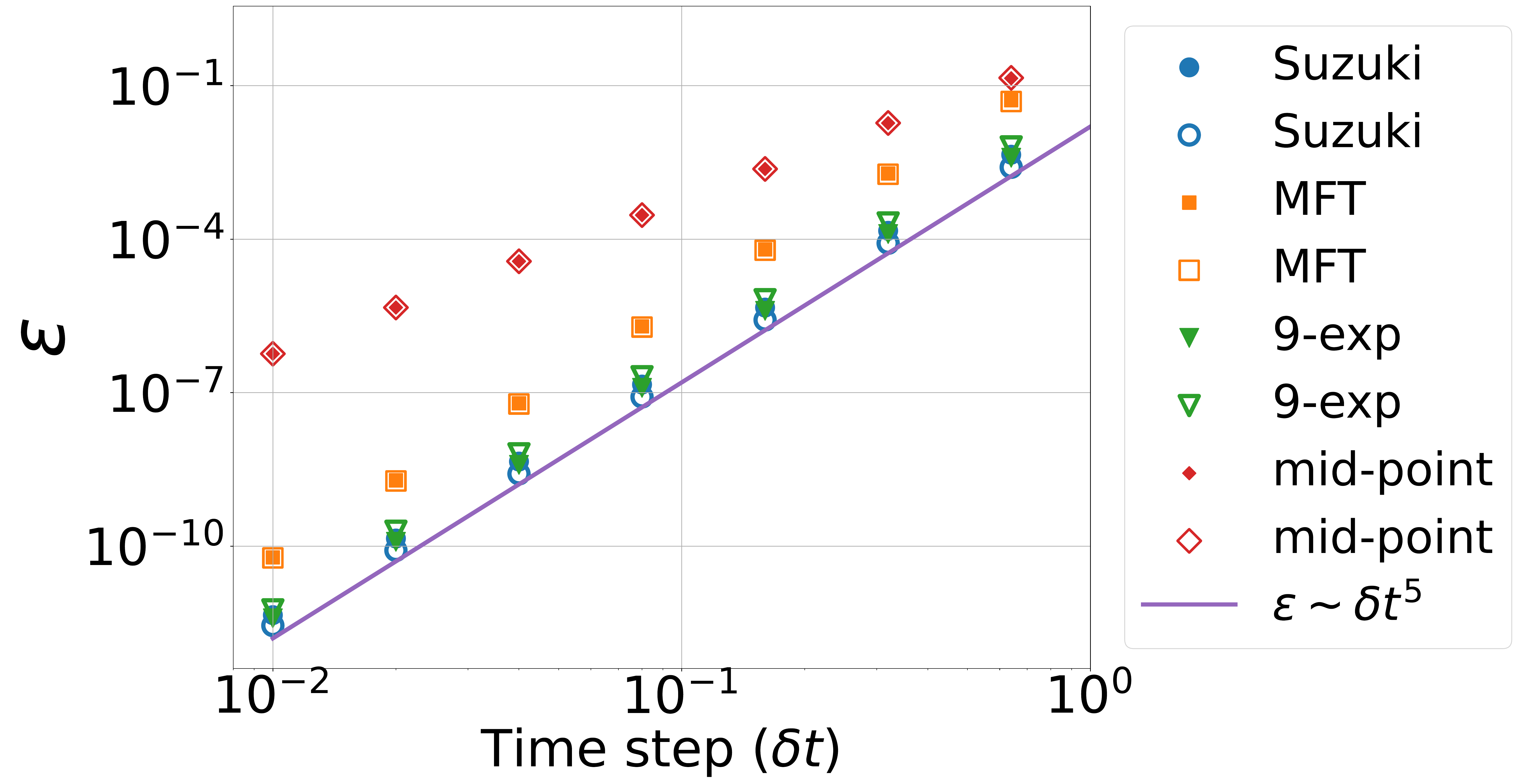}
\caption{Errors~\eqref{eq:error_def} of Trotterization formulas for the time-dependent Hamiltonian~\eqref{eq:Ht} at $\mu=1$.
Different symbols correspond to the fourth-order Suzuki (circle), MFT (square), 9-exponential (triangle), and mid-point rule (diamond) formulas.
The filled symbols correspond to the assignment $\sigma_x\to f(t)F$ and $t\sigma_z\to g(t)G$ in Eqs.~\eqref{eq:Hamiltonian} and \eqref{eq:Ht}, whereas the open symbols to the other \hl{assignment $\sigma_x\to g(t)G$ and $t\sigma_z\to f(t)F$ in these equations}.
The solid line guides the eye for the power law $\propto \delta t^5$.}
\label{fig:dt-dep}
\end{figure}

Now we implement each algorithm and numerically compare their accuracy.
To quantify the accuracy, we define the error $\epsilon$ as
\begin{align}
\epsilon =\left\| S\left(\mu+\frac{\delta t}{2},\mu-\frac{\delta t}{2}\right) - T \right\|,\label{eq:error_def}
\end{align}
where $T$ is the time-evolution unitary achieved through an algorithm and $\vert\vert \cdot \vert\vert$ is the Frobenius norm of the operator.
\hl{We have confirmed that the results do not change qualitatively when the spectral norm is used, as shown in Appendix~\ref{app:spec}.}
Figure~\ref{fig:dt-dep} shows the simulation errors for the seminal Landau-Zener Hamiltonian
\begin{align}
H(t)=\sigma_x + t \sigma_z,\label{eq:Ht}
\end{align}
confirming the expected error scaling ($\sigma_x$ and $\sigma_z$ are the Pauli matrices).
Here we fix $\mu=1$, for instance, and plot the errors of each formula against the time step $\delta t$.
Recall that we have two choices in assigning which of $\sigma_x$ and $t \sigma_z$ to $f(t)F$ (and the other to $g(t)G$).
In Fig.~\ref{fig:dt-dep}, the assignment $\sigma_x\to f(t)F$ and $t\sigma_z\to g(t)G$ is shown by filled symbols and the other \hl{assignment $\sigma_x\to g(t)G$ and $t\sigma_z\to f(t)F$} by open ones.
For both assignments, the mid-point rule gives $O(\delta t^3)$ errors, whereas the MFT, 9-exponential, and the fourth-order Suzuki formulas give $O(\delta t^5)$ errors, as expected.

For $\mu=1$ in Fig.~\ref{fig:dt-dep}, we observe that the error coefficient of the MFT is larger than that of the fourth-order Suzuki formula.
This is a penalty to pay with the MFT instead of its advantage of fewer exponentials.
However, it is noteworthy that the 9-exponential formula has error coefficients as small as Suzuki's, even though it still has fewer exponentials.
Here we recall that the errors of the MFT $\Lambda_5$ and the 9-exponential formula $\Lambda_5'$ consist of two parts as in Eqs.~\eqref{eq:MFT} and \eqref{eq:9exp}, including the common contribution $\Upsilon_5$ inherent to the time dependence.
Our observation about the difference in error coefficients derives rather from the time-independent formulas, $\Gamma_5$ and $\Gamma_5'$.

\begin{figure}
\includegraphics[width=\columnwidth]{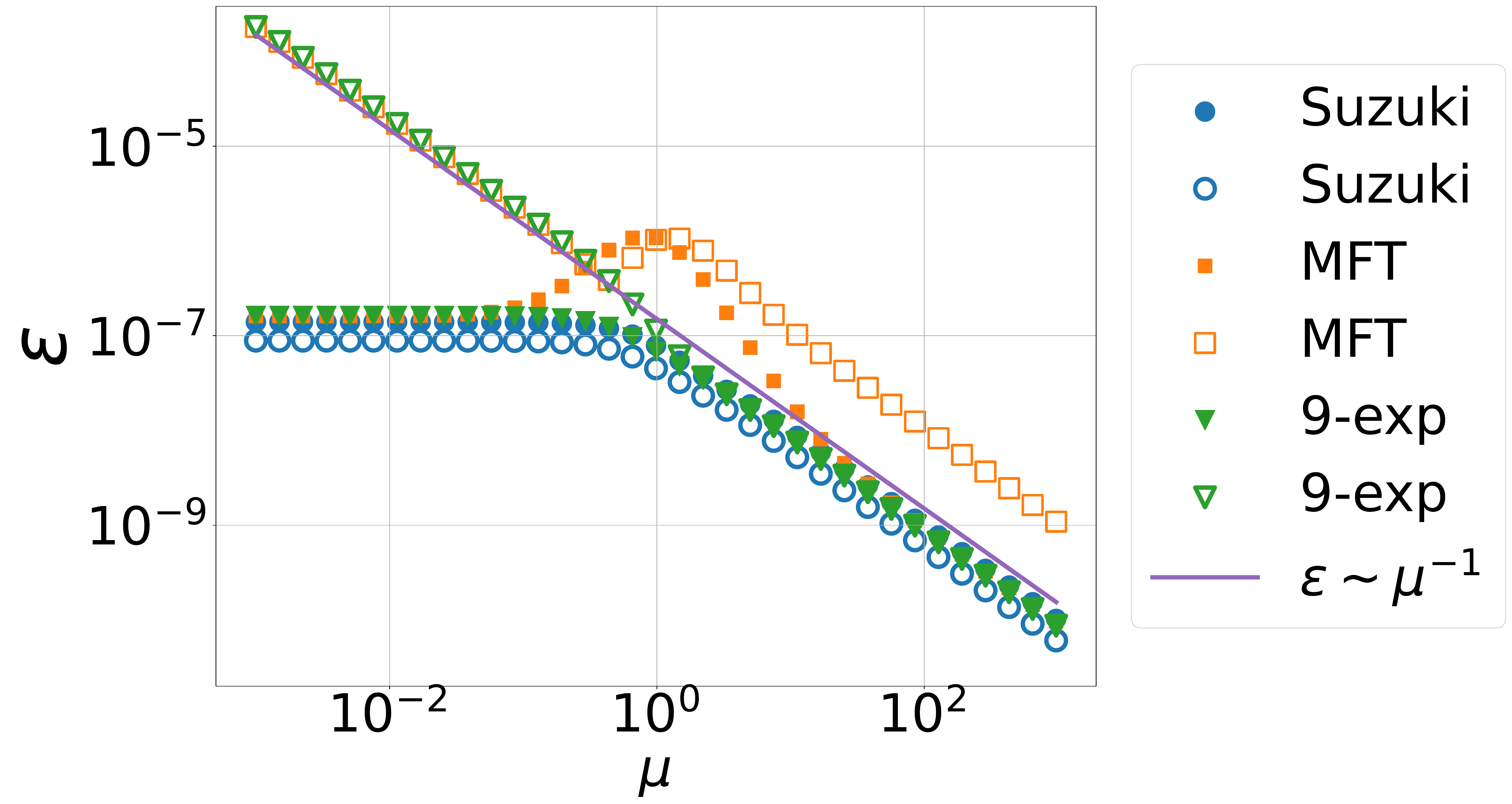}
\caption{Errors~\eqref{eq:error_def} of Trotterization formulas for the time-dependent Hamiltonian~\eqref{eq:Ht} for various $\mu$.
The time step $\delta t$ is taken, depending on $\mu$, as $\delta t = 0.1/\|H(\mu)\| = 0.1/\sqrt{1+\mu^2}$.
The symbols are the same as in Fig.~\ref{fig:dt-dep}: \hl{The filled symbols correspond to the assignment $\sigma_x\to f(t)F$ and $t\sigma_z\to g(t)G$ in Eqs.~\eqref{eq:Hamiltonian} and \eqref{eq:Ht}, whereas the open symbols to the other assignment $\sigma_x\to g(t)G$ and $t\sigma_z\to f(t)F$}.
The solid line guides the eye for the power law $\propto \mu^{-1}$.
}
\label{fig:mu-dep}
\end{figure}

To investigate the $\mu$-dependence systematically, we plot the errors in Fig.~\ref{fig:mu-dep}.
Here we set $\delta t = 0.1/\|H(\mu)\| = 0.1/\sqrt{1+\mu^2}$ for each $\mu$ since the magnitude of $H(t)$ significantly changes in such a wide range.
The figure shows the following two features.
First, as $\mu\to0$, the errors are divergent for the MFT and 9-exponential formulas in the assignment $\sigma_x\to g(t)G$ and $t\sigma_z\to f(t)F$.
This derives from the fact that, in that limit, $\beta_1=O(\delta t^2)$ and $\beta_{12}/\beta_1$.
Hence, to obtain a better approximation, we must choose the appropriate assignment $\sigma_x\to f(t)F$ and $t\sigma_z\to g(t)G$.
Second, the error of the 9-exponential formula, with the appropriate assignment, is as small as the fourth-order Suzuki's over the entire range of $\mu$, whereas that of the MFT can be as good as them in two limiting cases $\mu\to0$ and $\mu\to\infty$.
The worst case for the MFT is $\mu\sim1$, for which the $\delta t$-scaling is shown in Fig.~\ref{fig:dt-dep}.
We note again that its error comes dominantly from $\Gamma_5$ rather than the time-dependent part $\Upsilon_5$.
In fact, the $\mu^{-1}$-scaling and the difference between the two MFT results in the large $\mu$ regime can be understood by analyzing $\Gamma_5$ (see Appendix~\ref{app:error_detail} for details).

\section{Application to a quantum circuit}\label{sec:Ising}
\hl{Having studied the basic properties of the MFT and the 9-exponential formula in a single qubit model, we here apply them to a more practical model.
We consider a quantum Ising chain of length $L$ under an oscillatory transverse field, whose Hamiltonian is in the form of Eq.~\eqref{eq:Hamiltonian} with
\begin{align}
    &f(t) = \sin(t),\quad  F=\sum_{i=1}^L h_x\sigma_x^i,\label{eq:HamIsing1}\\
    &g(t) = 1, \quad G=\sum_{i=1}^L (J\sigma_z^i \sigma_z^{i+1}+h_z\sigma_z^i),\label{eq:HamIsing2}
\end{align}
where $\sigma_x^i$ and $\sigma_z^i$ are the Pauli matrices acting on the site $i$.
We impose the periodic boundary condition $\sigma_z^{L+1}=\sigma_z^{1}$ and set  $J=-1.0$, $h_z=0.2$, and $h_x=-2.0$.
}

\hl{
Let us compare the cost and accuracy of each algorithm by considering a time evolution from $t_i=0$ to $t_f=\pi$, for example.
For a number $N$ of steps and the step size $\delta t=(t_f-t_i)/N$, we introduce $t_k = t_i + k \delta t$ ($k=0,1,\dots,N$) to study the error 
\begin{align}
   \epsilon \equiv \left\|  S(t_f,t_i) - \prod_{k=1,\dots,N}^{\leftarrow} T(t_{k},t_{k-1}) \right\|,\label{eq:error_TFI}
\end{align}
where the Trotterization $T(t_{k},t_{k-1})$ is taken for either the mid-point rule, MFT, 9-exponential, or Suzuki formulas.
}

\hl{
We focus on the number of 1-qubit and 2-qubit gates, $N_\mathrm{gates}$, in each Trotterization.
For the mid-point rule, we have $T(t,s)=e^{-i\sin(\mu)\delta t F/2}e^{-i\delta t G}e^{-i\sin(\mu)\delta t F/2}$ ($\mu\equiv \frac{t+s}{2}$), which is further decomposed into 1- and 2-qubit gates by $e^{-i\sin(\mu)\delta t F/2}=\prod_{i=1}^L e^{-i\sin(\mu)\delta t \sigma_x^i/2}$ and $e^{-i\delta t G}=\prod_{i=1}^L e^{-iJ \sigma_z^i \sigma_z^{i+1}\delta t}\prod_{i=1}^L e^{-i h_z \sigma_z^i \delta t }$.
Thus, we have $N_\mathrm{gates}=5L N$ for the mid-point rule.
Likewise, we have $N_\mathrm{gates}=10L N$ for the MFT, $13LN$ for the 9-exponential formula, and $15LN$ for the Suzuki formula.
}

\begin{figure}
\includegraphics[width=\columnwidth]{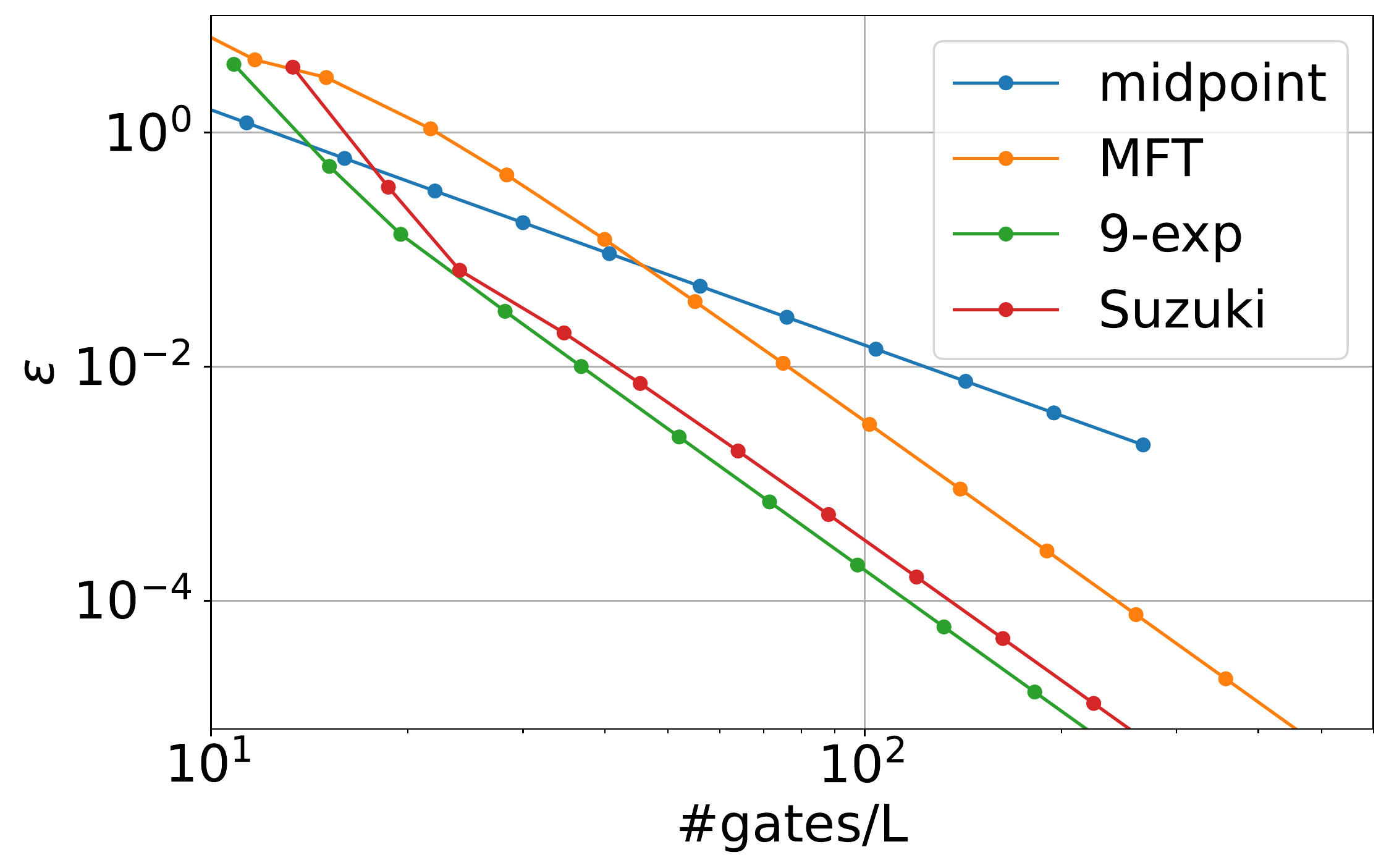}
\caption{\hl{Errors~\eqref{eq:error_TFI} of Trotterization formulas for the time-dependent quantum Ising Hamiltonian [Eqs.~\eqref{eq:HamIsing1} and \eqref{eq:HamIsing2}] at $L=6$ plotted against $N_\mathrm{gates}/L$.
The time evolution from $t_i=0$ to $t_f=\pi$ is calculated for steps $N$ ranging from 5 to 400, which is converted to $N_\mathrm{gate}$ for each Trotterization formula shown in the legend (see also text).}
}
\label{fig:error_vs_gates}
\end{figure}

\hl{
Figure~\ref{fig:error_vs_gates} shows the error~\eqref{eq:error_TFI} of each Trotterization calculated for several $N$ ranging from 5 to 400.
Here, the horizontal axis is taken as $N_\mathrm{gates}/L$, and the exact solution $S(t_f,t_i)$
was obtained by integrating with sufficiently small stepsize to guarantee asymptotic convergence below all other error scales.
We observe expected error scalings when $N_\mathrm{gates}\propto \delta t^{-1}$ increases; $\epsilon\propto \delta N_\mathrm{gates}^{-2}$ for the second-order mid-point rule and $\epsilon\propto N_\mathrm{gates}^{-4}$ for the other fourth-order formulas.
We note that these scaling exponents are larger by one than those for a single step because we successively operate $N\propto \delta t^{-1}\propto N_\mathrm{gates}$ steps.
}

\hl{
Importantly, the 9-exponential formula produces less error than the Suzuki formula at a given number of gates.
When compared with the same $N$, the Suzuki formula has less error (data not shown), but it requires $N_\mathrm{gates}=15LN$ gates that are more than $N_\mathrm{gates}=13LN$ of the 9-exponential formula.
This difference in the required number of gates results in the 9-exponential formula obtaining the best performance, as shown in Fig.~\ref{fig:error_vs_gates}
In contrast, the MFT gives larger errors than the other fourth-order formula, even though it requires fewer exponentials due to the large error coefficient in $\Gamma_5$.
}

\section{Discussion and Conclusion}
We studied the generalization of Trotterization formulas to time-dependent generators.
While the second-order formula is well-known, we derived a 7-exponential fourth-order formula $T_{7,4}$
\hl{and a 15-exponential sixth-order formula $T_{15,7}$} as in Eqs.~\eqref{eq:MFT} and \eqref{eq:MST}.
\hl{Remarkably, we showed that the number seven (fifteen) of exponentials are the minimum to make a fourth-order (sixth-order) formula, and hence we named them the minimum fourth-order and sixth-order formulas (MFT and MST).
Our approach can be generalized to higher-order formulas.}
An important application of Trotterization formulas is quantum dynamics simulation on digital quantum computers, where reducing the number of exponentials (i.e., quantum gates) is beneficial for suppressing errors accompanied by each gate operation.
Although Trotterization-based algorithms do not have optimal error scalings of quantum-singular-value-transformation-based ones~\cite{Low2017,Martyn2021,Watkins2022,Mizuta2023}, they have fewer overheads and are useful in near-term quantum computers. 
As demonstrated in a single-qubit model, the MFT's error is as tiny as Suzuki's fourth-order formula, even though the MFT involves fewer exponentials.
\hl{Also, the 9-exponential formula has less error per gate than Suzuki's formula, as demonstrated in a many-spin model.
These formulas would also be useful in MPS-based dynamics simulations, where fewer exponentials could speed up the simulation.
In addition, the fourth-order formulas can be used for estimating the errors of lower-order Trotterizations~\cite{Ikeda2023}.}

Our result is limited in two ways.
First, we assumed Eq.~\eqref{eq: time dependence} and excluded more general time dependence $A(t)=X(t)+Y(t)$.
Equation~\eqref{eq: time dependence} includes important applications like the adiabatic state preparation, but it would be desirable if one could generalize the assumption.
Second, our result is currently limited to the decomposition into two parts $X$ and $Y$.
In some physically relevant models, such as the quantum Ising model with external fields and the Heisenberg model, the Trotterization for the two operators is sufficient since those Hamiltonians can be viewed as a sum of two operator sets, each of which consists of commuting operators.
However, it would be useful if one could extend the MFT to more than two operators.

The error of the MFT was obtained as $\Lambda_5$ in Eq.~\eqref{eq:MFT}
consisting of time-independent and -dependent contributions.
For a quantum many-body local Hamiltonian on a lattice, both contributions are proportional to the volume since they consist of (nested) commutators between local Hamiltonian parts.
If we use some norm for the error operator as an error estimator, it will also be proportional to the volume.
However, one is usually interested in the evolution of a given vector $v$, and the error is evaluated as $\| \Lambda_5 v\|$.
Such an idea was proposed to improve the error estimates by Jahnke et al.~\cite{Jahnke2000} for time-independent cases and generalized by An et al.~\cite{An2021} to the time-dependent Suzuki formulas.
A similar idea could be used to estimate the error of the MFT.
Also, errors could be reduced by making the time step $\delta t$ adaptive~\cite{Zhao2022,Zhao2023}.
We leave further investigations of the time-dependent error $\Upsilon_5$ for future work.

\section*{Acknowledgements}
The authors thank Hongzheng Zhao and Kaoru Mizuta for helpful discussions.
T. N. I. was supported by JST PRESTO Grant No. JPMJPR2112 and by JSPS KAKENHI Grant No. JP21K13852. 
A.A. and S.S. wish to thank NTT Research for their financial and technical support. 

\bibliographystyle{plainnat}
\bibliography{references}

\onecolumn 
\appendix

\section{Magnus expansion and continuous BCH formula}\label{app:Magnus}
The Magnus expansion~\cite{Blanes2009} up to the fourth order in our notation~\eqref{eq:Magnus} reads
\begin{align}
\Omega_1 &= \int_{\mu-\delta t/2}^{\mu+\delta t/2}ds_1 A(s_1),\\
\Omega_2 &= \frac{1}{2}\int_{\mu-\delta t/2}^{\mu+\delta t/2}ds_1\int_{\mu-\delta t/2}^{s_1}ds_2 [A(s_1),A(s_2)],\\
\Omega_3 &= \frac{1}{6}\int_{\mu-\delta t/2}^{\mu+\delta t/2}ds_1\int_{\mu-\delta t/2}^{s_1}ds_2 \int_{\mu-\delta t/2}^{s_2}ds_3 ([A(s_1),[A(s_2),A(s_3)]]+[[A(s_1),A(s_2)],A(s_3)]),\\
\Omega_4 &= \frac{1}{12}\int_{\mu-\delta t/2}^{\mu+\delta t/2}ds_1\int_{\mu-\delta t/2}^{s_1}ds_2 \int_{\mu-\delta t/2}^{s_2}ds_3 \int_{\mu-\delta t/2}^{s_3} ds_4 ([[[A(s_1),A(s_2)],A(s_3)],A(s_4)]\notag\\
&\qquad\qquad+[A(s_1),[[A(s_2),A(s_3)],A(s_4)]]]+[A(s_1),[A(s_2),[A(s_3),A(s_4)]]]\notag\\
&\qquad\qquad\qquad+[A(s_2),[A(s_3),[A(s_4),A(s_1)]]]).
\end{align}
For $A(t)=u_1(t)Z_1+u_2(t)Z_2$, $\Omega_n$ reduce to
\begin{align}
&\Omega_1 = \omega_1 Z_1 + \omega_2 Z_2,\
\Omega_2 = \frac{1}{2}(\omega_{21}-\omega_{12})[Z_1,Z_2],\
\Omega_3 = \sum_{i=1}^2 \beta_{i12}[Z_i,[Z_1,Z_2]],\notag\\
&\Omega_4 = \sum_{i=1}^2\sum_{j=1}^2 \beta_{ij12}[Z_i,[Z_j,[Z_1,Z_2]]],
\end{align}
with
\begin{align}
\beta_i &= \omega_i \quad (i=1,2),\label{eq:beta_i}\\
\beta_{12} &= \frac{1}{2}(\omega_{21}-\omega_{12}),\label{eq:beta_12}\\
\beta_{i12} &= \frac{1}{6}(\omega_{21i}-\omega_{12i}-\omega_{i21}+\omega_{i12}),\\
\beta_{ij12} &=-\frac{1}{12}(\omega_{ij21}-\omega_{ij12}+\omega_{j12i}-\omega_{j21i}
+\omega_{21ji}-\omega_{12ji}+\omega_{1ji2}-\omega_{2ji1}),
\end{align}
where we introduced
\begin{align}
\omega_{i_1\cdots i_S} \equiv \int_{\mu-\delta t/2}^{\mu+\delta t/2}dt_S \int_{\mu-\delta t/2}^{t_S}dt_{S-1}\dots \int_{\mu-\delta t/2}^{t_3}dt_2\int_{\mu-\delta t/2}^{t_2}dt_1 u_{i_S}(t_S)\dots u_{i_1}(t_1).
\end{align}

\section{Orthogonal polynomial expansions of continuous BCH formula}\label{app:poly}
As shown in Appendix~\ref{app:Magnus}, the continuous BCH formula is characterized by the coefficients $\beta_j$, $\beta_{12}$, $\beta_{i12}$, $\beta_{ij12}$, etc.
Here, considering $\delta t \to 0$, we derive their leading-order contributions in $\delta t$, using orthogonal polynomials.

We introduce the orthogonal polynomials $P_n(x)$ $(n=0,1,2,\dots)$ with $\mathrm{deg}P_n=n$ on domain $-1\le x\le 1$ that satisfy
\begin{align}
\int_{-1}^1 P_{m}(x)P_n(x) W(x) dx  = c_n \delta_{mn},
\end{align}
where $W(x)$ is the weight function and $c_n$ are real numbers.
For example, the Legendre polynomials correspond to $W(x)=1$ and $c_n=2/(2n+1)$.
Using this, we expand $u_i(t)$ as
\begin{align}\label{eq:poly_expansion}
u_i(\mu+s) = \frac{1}{\delta t} \sum_{n=1}^\infty u_i^{(n)}(\mu)P_{n-1}\left(\frac{s}{\delta t/2}\right),
\end{align}
which implies
\begin{align}\label{eq:polycoeff}
u_i^{(n)}(\mu) =\frac{\delta t}{c_{n-1}}\int_{-1}^1 dv u_i(\mu+v\delta t/2)P_{n-1}(v)W(v).
\end{align}
Using the Taylor expansion for $u_i$ in Eq.~\eqref{eq:polycoeff}, we obtain $u_i^{(n)}(\mu)=O(\delta t^n)$.
Based on Eq.~\eqref{eq:poly_expansion}, we can obtain the leading-order contributions of $\beta$'s in Appendix~\ref{app:Magnus}.

For example, we obtain, for the Legendre polynomials,
\begin{align}
\beta_i &= u_i^{(1)}=\frac{\delta t}{2}\int_{-1}^1 dv u_i(\mu+v\delta t/2),\\
\beta_{12} &= \frac{2}{3}( u_2^{(1)}u_1^{(2)}- u_1^{(1)}u_2^{(2)})+O(\delta t^5).
\end{align}
We note that $\beta_i$ consists of odd-order terms of $\delta t$
\begin{align}
\beta_i = u_i(\mu)\delta t+\frac{u_i''(\mu)}{24}\delta t^3+O(\delta t^5).
\end{align}
Thus, if $\beta_i=o(\delta t)$, then $\beta_i=O(\delta t^3)$.

\section{\hl{Derivation of the MST}}\label{app:higher}
Here we derive the minimum sixth-order Trotterization formula in Theorem~\ref{thm:MST}.
To begin with, we invoke the minimum sixth-order formula for time-independent operators~\cite{Yoshida1990}, having
\begin{align}
    e^{\beta_1 X +\beta_2Y}
    &= e^{a_1 \beta_1X}e^{b_1 \beta_2 Y} e^{a_2 \beta_1X}e^{b_2 \beta_2 Y} e^{a_3 \beta_1X}e^{b_3 \beta_2 Y} e^{a_4 \beta_1X}e^{b_4 \beta_2Y}\notag\\
    &\qquad \times e^{a_4 \beta_1 X} e^{b_3 \beta_2 Y} e^{a_3 \beta_1X}e^{b_2 \beta_2 Y} e^{a_2 \beta_1X}e^{b_1 \beta_2 Y}e^{a_1 \beta_1X}
    +O(\delta t^7),\label{eq:Yoshida}
\end{align}
where
\begin{align}\label{eq:abvals_Yoshida}
\begin{array}{ll}
a_1=0.39225680523878, & b_1=0.78451361047756, \\
a_2=0.5100434119184585, & b_2=0.235573213359357, \\
a_3=-0.4710533854097566, & b_3=-1.17767998417887, \\
a_4=\frac{1}{2}-\sum_{i=1}^3 a_i=0.0687531682525181, & b_4=1-2 \sum_{i=1}^3 b_i=1.31518632068391.
\end{array}
\end{align}
Note that this Yoshida formula is known to have a large error coefficient~\cite{Ostmeyer2023}.
To reduce the coefficient with adding more exponential, one can use the Blanes and Moan formula~\cite{Blanes2002}, like we derived the 9-exponential formula in the main text.

Now we consider its time-dependent generalization, for which 
\begin{align}
    S\left(\mu+\frac{\delta t}{2},\mu-\frac{\delta t}{2}\right)
    &=\exp\left( \beta_1 X +\beta_2 Y +\beta_{12}[X,Y]+\beta_{112}[X,[X,Y]]+\beta_{212}[Y,[X,Y]]\right.\notag\\
    &\qquad\qquad \left. +
    \beta_{1112}[X,[X,[X,Y]]]
+\beta_{2212}[Y,[Y,[X,Y]]] \right.\notag\\
&\qquad\qquad\qquad \left. +(\beta_{1212}+\beta_{2112})[X,[Y,[X,Y]]]\right) +O(\delta t^7),\label{eq:Magnus6}
\end{align}
according to the Magnus expansion, where $\beta_{12}=O(\delta t^3)$, and $\beta_{i12}$ and $\beta_{ij12}$ are $O(\delta t^5)$.
We aim to decorate the right-hand side of Eq.~\eqref{eq:Yoshida} so that the left-hand side coincides with Eq.~\eqref{eq:Magnus6}.

To reproduce the six commutator terms in the exponent of Eq.~\eqref{eq:Magnus6}, we introduce six unknown variables to consider
\begin{align}
\Phi&=
    e^{(a_1 \beta_1+u_4)X}e^{(b_1 \beta_2 +u_3)Y}e^{(a_2 \beta_1+u_2) X}e^{(b_2 \beta_2 +u_1 -z)Y} e^{(a_3 \beta_1-w) X}
 e^{(b_3 \beta_2 +z)Y}e^{(a_4 \beta_1 +w) X}e^{b_4 \beta_2 Y}\notag\\
 &\qquad\quad \times e^{(a_4 \beta_1+w) X}e^{(b_3 \beta_2 +z)Y}e^{(a_3 \beta_1 -w)X} e^{(b_2 \beta_2-u_1-z) Y}\times e^{(a_2 \beta_1-u_2) X}e^{(b_1 \beta_2 -u_3)Y}e^{(a_1 \beta_1-u_4)X},\label{eq:Phidef}
\end{align}
where we impose as working hypotheses that $u_1,\dots,u_4$ are $O(\delta t^2)$, and $w$ and $z$ are $O(\delta t^3)$.
Applying BCH-type formulas repeatedly and using Eq.~\eqref{eq:Yoshida}, we obtain
\begin{align}
    \Phi
    &=\exp\left( \beta_1 X +\beta_2 Y +c_{12}[X,Y]+c_{112}[X,[X,Y]]+c_{212}[Y,[X,Y]]\right.\notag\\
    &\qquad \left. +
    c_{1112}[X,[X,[X,Y]]]
+c_{2212}[Y,[Y,[X,Y]]]
+c_{1212}[X,[Y,[X,Y]]]
\right)
+O(\delta t^7),\label{eq:PhiBCH}
\end{align}
where $c_{12},c_{112},\dots,c_{1212}$ are second-order polynomials in terms of the unknowns $u_1,\dots,u_4,w,$ and $z$.
Using Eq.~\eqref{eq:abvals_Yoshida}, these polynomials are given by
\begin{align}
c_{12} &=0.804600434314477 u_1 \beta _1-0.21548638952244 u_3 \beta _1-0.56902722095512 u_2 \beta _2+ u_4 \beta _2,\\
c_{112}&=-0.28451361047756 \beta _2 u_2^2+0.804600434314477 u_1 \beta _1 u_2-0.56902722095512 u_4 \beta _2 u_2 \notag\\
&\qquad -0.157118466580002 z \beta _1^2+0.804600434314477 u_1 u_4 \beta
   _1-0.21548638952244 u_3 u_4 \beta _1\notag\\
   &\qquad\qquad +0.5 u_4^2 \beta _2 -0.161938460199746 w \beta _1 \beta _2,\\
c_{212}&=0.402300217157238 \beta _1 u_1^2+0.804600434314477 u_3 \beta _1 u_1-0.161938460199745 w \beta _2^2\notag\\
&\qquad -0.10774319476122 u_3^2 \beta _1-0.56902722095512 u_2 u_3 \beta _2-0.489977318150775
   z \beta _1 \beta _2,\\
c_{1112} &=
-0.0118215295615413 u_1 \beta _1^3+0.00856168382290096 u_3 \beta _1^3+0.0562690326323137 u_2 \beta _2 \beta _1^2,\\
c_{2212}&=0.0160325321433039 u_2 \beta _2^3+0.0641595078732893 u_1 \beta _1 \beta _2^2+0.065376134206464 u_3 \beta _1 \beta _2^2,\\
c_{1212}&=
0.0115567664079044 u_1 \beta _2 \beta _1^2+0.0538195677848599 u_3 \beta _2 \beta _1^2+0.112538065264628 u_2 \beta _2^2 \beta _1.
\end{align}

Thus our problem is reduced to solving the following set of six equations,
\begin{align}
    &c_{12}=\beta_{12},\quad c_{1112} = \beta_{1112}, \quad c_{2212} = \beta_{2212}, \quad c_{1212} = \beta_{1212} + \beta_{2112},\label{eq:ceq1}\\
    & c_{112} = \beta_{112},\quad c_{212} = \beta_{212},\label{eq:ceq2}
\end{align}
for $u_1,\dots,u_4,w,$ and $z$.
Due to the linearity and independence, Eqs.~\eqref{eq:ceq1} can be uniquely solved for $u_1,\dots,u_4$ under our assumptions $\beta_1=\Theta(\delta t)$ and $\beta_2=\Theta(\delta t)$.
Then, substituting these solutions to Eqs.~\eqref{eq:ceq2}, we have two linear and independent equations for $w$ and $z$, which can also be uniquely solved.
Thus we have proved that $\Phi$ with these solutions coincides with the exact propagator $S(\mu-\delta t/2,\mu+\delta t/2)$ up to the sixth order of $\delta t$.

\section{\hl{Spectral norm}}\label{app:spec}
In the main text, the Frobenius norm in the error definition~\eqref{eq:error_def}.
To emphasize this norm choice, we rewrite Eq.~\eqref{eq:error_def} as 
\begin{align}
\epsilon_F =\left\| S\left(\mu+\frac{\delta t}{2},\mu-\frac{\delta t}{2}\right) - T \right\|_F,\label{eq:error_def_app}
\end{align}
where $\| A \| = \sqrt{\mathrm{tr}(A^\dag A)}$.
Note that $2\epsilon_F$ represents the trace distance of $S$ and $T$.
We can also consider the spectral norm 
\begin{align}
\epsilon_S =\left\| S\left(\mu+\frac{\delta t}{2},\mu-\frac{\delta t}{2}\right) - T \right\|_S,\label{eq:error_def_spec}
\end{align}
where $\| A\|_S$ coincides with the largest singular value of $A$.
We confirm that both choices give qualitatively similar results in the benchmark calculation presented in Fig.~\ref{fig:dt-dep}.
We illustrate, in Fig.~\ref{fig:ratio}, the error ratio obtained for the same simulation,
where we observe that the different choice results only in a constant factor independent of the algorithm and time step.

\begin{figure}
\begin{center}
  \includegraphics[width=10cm]{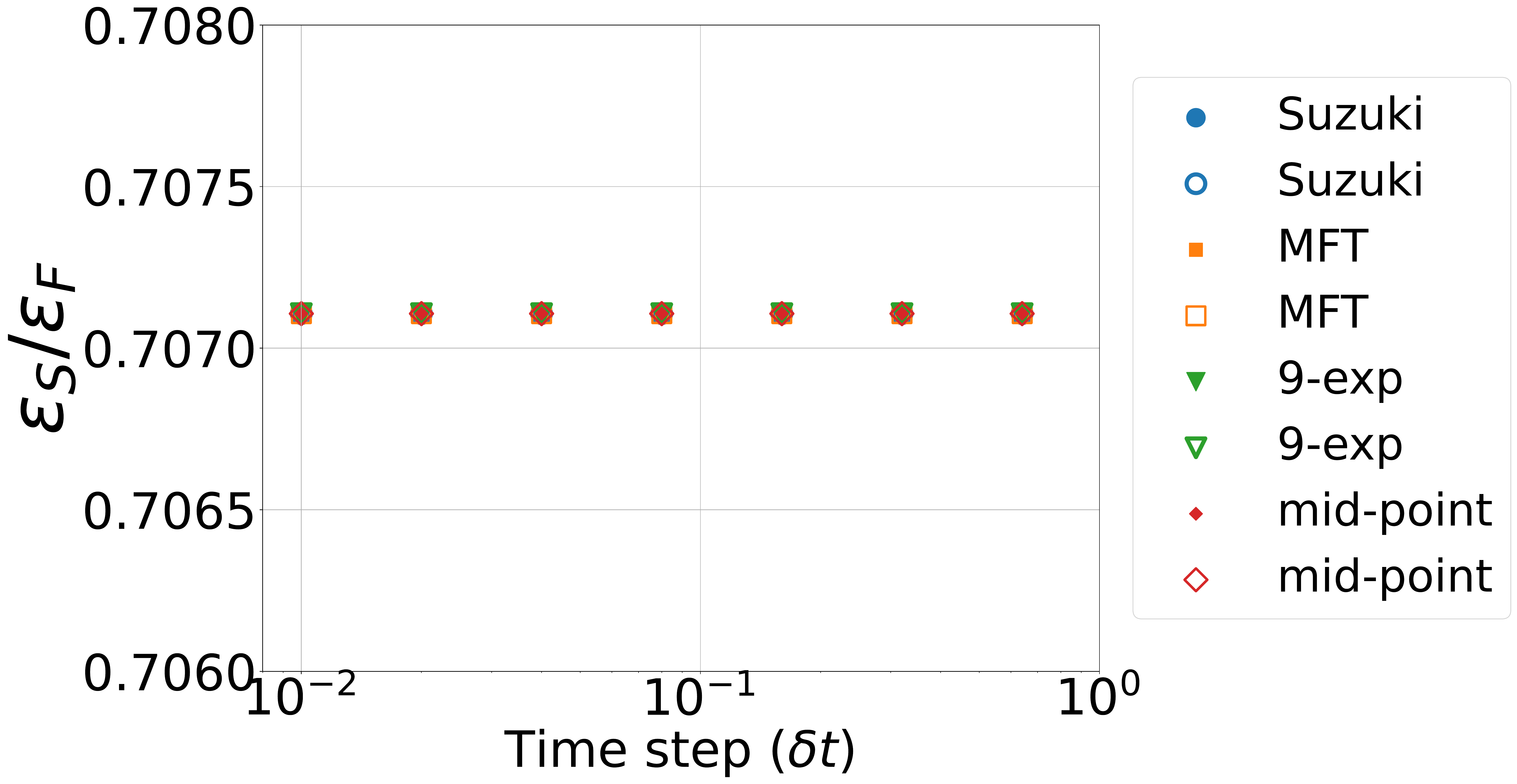}
\caption{Error ratio between those evaluated by the spectral and Frobenius norm for Trotterization formulas.
The parameters and legends are all shared by those in Fig.~\ref{fig:dt-dep}.}
\label{fig:ratio}
\end{center}
\end{figure}

\section{Details on $\Gamma_5$ and $\Gamma_5'$}\label{app:error_detail}
Here we supplement more details on the time-independent errors $\Gamma_5$ and $\Gamma_5'$.
According to Ref.~\cite{Omelyan2002}, each symmetric Trotterization formula has the following relation,
\begin{align}
\mathrm{e}^{(A+B) h+C_1 h+C_3 h^3+C_5 h^5+\cdots}=\mathrm{e}^{A a_1 h} \mathrm{e}^{B b_1 h} \cdots \mathrm{e}^{B b_q h} \mathrm{e}^{A a_{q+1} h},
\end{align}
where $h$ is the time step and
\begin{align}
    \begin{aligned}
C_1= & (\nu-1) A+(\sigma-1) B \\
C_3= & \alpha[A,[A, B]]+\beta[B,[A, B]] \\
C_5= & \gamma_1[A,[A,[A,[A, B]]]]+\gamma_2[A,[A,[B,[A, B]]]]+\gamma_3[B,[A,[A,[A, B]]]] \\
& +\gamma_4[B,[B,[B,[A, B]]]]+\gamma_5[B,[B,[A,[A, B]]]]+\gamma_6[A,[B,[B,[A, B]]]].
\end{aligned}
\end{align}
Each fourth-order formula is obtained by choosing $a_i$'s and $b_i$'s so that $C_1=C_3=0$, in which case the Trotter error is proportional to $C_5$ in the leading order:
\begin{align}
\mathrm{e}^{(A+B)h}-\mathrm{e}^{A a_1 h} \mathrm{e}^{B b_1 h} \cdots \mathrm{e}^{B b_q h} \mathrm{e}^{A a_{q+1} h}=
h^5 C_5 +O(h^7).
\end{align}
We have $q=3$ for the Forest-Ruth-Suzuki formula and $q=4$ for Omelyan's Forest-Ruth formula, for which we have different values of $\gamma_1,\dots,\gamma_6$.
Omelyan's Forest-Ruth formula has more degrees of freedom that are fine-tuned to decrease an error quantifier $\sum_{i=1}^6 |\gamma_i|^6$~\cite{Omelyan2002}.
The time-independent contributions $\Gamma_5$ and $\Gamma_5'$ are given, in the leading order of $\delta t$, by $h^5C_5$ under the substitution of $Ah\to \beta_1 X$ and $Bh\to \beta_2 Y$ (or the other assignment $Ah\to \beta_2 Y$ and $Bh\to \beta_1 X$).

Now we explain why the filled and open symbols for the MFT in Fig.~\ref{fig:mu-dep} deviate significantly in the large $\mu$ region.
We recall that
\begin{align}
    &\beta_1= \delta t,\ F=\sigma_x,\ \beta_2= \mu\delta t,\ G=\sigma_z, \quad (\text{for}\ \blacksquare),\label{eq:fillsq}\\
    &\beta_1= \mu \delta t,\ F=\sigma_z,\ \beta_2= \delta t,\ G=\sigma_x, \quad (\text{for}\ \square),\label{eq:opensq}
\end{align}
where we used $\int_{\mu-\delta t/2}^{\mu+\delta t/2} f(t)dt=\delta t$ and $\int_{\mu-\delta t/2}^{\mu+\delta t/2} g(t)dt=\mu\delta t$.
In Fig.~\ref{fig:mu-dep}, we set $\delta t=0.1/\sqrt{1+\mu^2}$ so $\beta_1\sim 0.1/\mu$ and $\beta_2\sim 0.1$ for $\blacksquare$ whereas $\beta_1\sim 0.1$ and $\beta_2\sim 0.1/\mu$ for $\square$ when $\mu\gg1$.
This asymmetry gives rise to the difference in the error of the MFTs.
As discussed in the main text, the dominant contribution of the error comes from the time-independent part $\Gamma_5$, which is $h^5 C_5$ in the present notation.
Therefore, the errors in the leading order are obtained by substituting $Ah\to \beta_1 (-iF) $ and $Bh\to \beta_2 (-iG)$ as
\begin{align}
    \Gamma_5 \sim
    \begin{cases}
        (-i)^5 \gamma_4 \frac{0.1^5}{\mu}[\sigma_z,[\sigma_z,[\sigma_z,[\sigma_x,\sigma_z]]]] & \text{for}\ \blacksquare,\ \mu\gg1\\
        (-i)^5 \gamma_1 \frac{0.1^5}{\mu}[\sigma_z,[\sigma_z,[\sigma_z,[\sigma_x,\sigma_z]]]] & \text{for}\ \square,\ \mu\gg1
    \end{cases},
\end{align}
where we used Eqs.~\eqref{eq:fillsq} and \eqref{eq:opensq} and the asymptotic behaviors of $\beta_1$ and $\beta_2$ in $\mu\gg1$.
Note that we obtain the observed $\mu^{-1}$ scaling for both cases.
Interestingly, the difference between these cases comes from the imbalance of $\gamma_4$ and $\gamma_1$, which are~\cite{Omelyan2002}
\begin{align}
    \gamma_1 = -0.0004138,\quad
    \gamma_4 = +0.0046844,
\end{align}
and thus differ by a factor of 10.
This is why we observed a magnitude difference between $\blacksquare$ and $\square$ in Fig.~\ref{fig:mu-dep}.

For generic values of $\mu$, $\mu$-dependence of $\delta t$ and more than one coefficients $\gamma_1,\cdots,\gamma_6$ come into play.
This may cause a peak near $\mu\sim1$ in Fig.~\ref{fig:mu-dep}, but we do not go into more detail since this should depend on the model and parameters.


\end{document}